\documentclass[11pt]{amsart}

\usepackage[utf8]{inputenc}
\usepackage{amsmath, amssymb, amsthm}
\usepackage{mathtools}
\usepackage{hyperref}
\usepackage{enumitem}

\usepackage[a4paper, margin=3.5cm]{geometry}
\theoremstyle{plain}
\newtheorem{theorem}{Theorem}[section]
\newtheorem{proposition}[theorem]{Proposition}

\newtheorem{corollary}[theorem]{Corollary}

\theoremstyle{definition}
\newtheorem{definition}[theorem]{Definition}
\newtheorem{example}[theorem]{Example}
\newtheorem{remark}[theorem]{Remark}

\newcommand{\F}{\mathbb{F}}
\newcommand{\R}{\mathbb{R}}
\newcommand{\Z}{\mathbb{Z}}
\newcommand{\Fq}{\F_q}
\newcommand{\Fqs}{\F_q^{\!\star}}
\newcommand{\Fqn}{\F_q^n}
\newcommand{\Fqns}{\Fqn\setminus\{0\}}
\newcommand{\wt}{\operatorname{wt}_H}
\newcommand{\dH}{d_H}
\newcommand{\angH}{\operatorname{angle}_H}
\newcommand{\angE}{\operatorname{angle}_E}

\title[Angle Between Vectors over Finite Fields]{Angle Between Two Vectors over Finite Fields\\
and an Application to Projective Unique Decoding}

\author{Kamil OTAL
}

\address{T\"UB\.ITAK B\.ILGEM UEKAE, Gebze, Kocaeli, T\"urkiye}
\email{kamil.otal@gmail.com}

\date{\today}
\begin{document}

\begin{abstract}
We introduce a Hamming-type angular function
$$\angH(u,v) := \min_{c \in \Fqs} \dH(u, cv)$$
on pairs of nonzero vectors in $\Fqn$ and show that it satisfies all
three metric axioms up to scalar multiplication. The function $\angH$
is invariant under nonzero scalar multiplication in either argument
and therefore descends to a genuine integer-valued metric on the
projective space $\mathbb{P}(\Fqn)$. As a concrete application, we
prove an \emph{angular} (or \emph{projective}) version of the
unique-decoding theorem for linear codes: if $\angH(u,C\setminus\{0\}) < d/2$,
where $d$ is the minimum distance of the linear code $C$, then the
closest direction in $C$ to $u$ is unique up to nonzero scalar
multiplication. We then discuss how this angular viewpoint relates
to the proximity-gap programme for Reed--Solomon codes. To the best of our knowledge, this is the first attempt to define an angle notion for vectors over finite fields and interpret it from several perspectives, including geometry, coding theory, and cryptography. 
\end{abstract}

\maketitle

\section{Introduction}

Distance and angle are the two basic geometric primitives that
underlie the metric viewpoint on a vector space. Over $\R^n$, both
admit elegant definitions in terms of an inner product, and the
\emph{angular metric}
$$\angE(u,v) = \arccos\!\left(\frac{\langle u,v\rangle}{\|u\|\,\|v\|}\right),
\qquad u,v \in \R^n\setminus\{0\},$$
is the cornerstone of any quantitative notion of ``parallelism''
between vectors.

Over a finite field $\Fq$, the analog of the Euclidean inner product
is well known to be poorly behaved as a measurement of geometry: in
particular, one may have $u\cdot u = 0$ for a nonzero vector $u$, and
the relation $u\cdot v = 0$ does not faithfully capture
orthogonality in any geometric sense. Distance, on the other hand,
is well behaved: the Hamming distance $\dH$ on $\Fqn$ is a genuine
metric, and it is the engine of virtually all coding theory.

In this note we observe that the Hamming distance can be used to
\emph{define} a satisfactory angular notion on $\Fqn$, without ever
invoking an inner product. Explicitly, for nonzero $u, v \in \Fqn$ we
set
$$\angH(u, v) := \min_{c \in \Fqs} \dH(u, cv).$$
Our main theorem (Theorem~\ref{thm:main}) asserts that $\angH$ enjoys
the three identifying properties of a metric, up to scalar
multiplication. Equivalently, $\angH$ descends to an
integer-valued metric on the projective space $\mathbb{P}(\Fqn)$.

The motivation for studying such an object is twofold.
\emph{Conceptually}, it furnishes a geometric vocabulary --- a notion
of two vectors being ``more or less parallel'' --- in a setting where
the inner-product picture fails. \emph{Practically}, the quantity
$\angH$ is exactly the Hamming distance from a vector to a punctured
line through the origin, an object that arises naturally in proximity
testing for linear codes \cite{BCKS25,GHKSS25,GBP25} and in the
list-decoding literature \cite{CHK25}.

Our contributions are as follows.
\begin{enumerate}[label=(\arabic*)]
\item We give a self-contained proof that $\angH$ satisfies the
metric axioms up to scalar equivalence (Theorem~\ref{thm:main}).
\item We record the projective-space interpretation, give a linear-time
algorithm for computing $\angH$, and contrast it with the Euclidean
angular metric (Section~\ref{sec:properties}).
\item We define the Hamming angle from a vector to a set, carefully
relating $\angH(u, S)$ and $\dH(u, S)$ when $S$ is a linear code
(Section~\ref{sec:application}, Proposition~\ref{prop:angle-vs-dist}).
\item As a concrete application, we prove an \emph{angular unique
decoding theorem}: a linear code $C$ of minimum distance $d$ admits
unique projective decoding up to angular radius $d/2$
(Theorem~\ref{thm:angular-unique-decoding}). We then discuss the
relationship between this angular viewpoint and the proximity-gap
problem.
\end{enumerate}

\paragraph{Notation.}
Throughout, $q$ is a prime power and $\Fq$ is the finite field of $q$
elements. We write $\Fqs := \Fq\setminus\{0\}$ for the multiplicative
group. For $u \in \Fqn$, $\wt(u) := |\{i \in [n] : u_i \neq 0\}|$
denotes the Hamming weight, and
$\dH(u,v) := \wt(u-v) = |\{i : u_i\neq v_i\}|$
denotes the Hamming distance. Note the basic identity
$\wt(\alpha u) = \wt(u)$ for any $\alpha \in \Fqs$, which we use
repeatedly without further comment.

\section{Angle in the Euclidean Setting}\label{sec:euclidean}

We briefly revisit the Euclidean situation to motivate what follows.
The signed (oriented) angle between two vectors, although familiar in
the plane, fails as a metric for the standard reasons: it is defined
only modulo $2\pi$, can be negative, and is antisymmetric in the
sense that $\theta(u,v) = -\theta(v,u)$. The triangle inequality
also fails in general. Worse, in dimensions $\geq 3$ there is no
canonical orientation, and the signed angle is not even uniquely
defined without additional choices.

The remedy is to pass to the \emph{unsigned} angle, which is the
length of the shortest arc on the unit sphere connecting the two
normalized vectors. Let $S^{n-1} = \{u \in \R^n : \|u\| = 1\}$ and
$d_{\mathrm{geo}}(u,v) := \arccos\langle u,v\rangle$ denote the
geodesic distance on $S^{n-1}$. Then $d_{\mathrm{geo}}$ is a metric
on $S^{n-1}$ (see, e.g., \cite[Chapter~3]{doCarmo}). For arbitrary
nonzero $u,v \in \R^n$, one defines
$$\angE(u,v) := d_{\mathrm{geo}}\!\left(\frac{u}{\|u\|}, \frac{v}{\|v\|}\right)
= \arccos\!\left(\frac{\langle u,v\rangle}{\|u\|\,\|v\|}\right),$$
which is invariant under positive scalar multiplication in either
argument. Concretely, $\angE$ satisfies:
\begin{enumerate}[label=(\arabic*)]
\item $\angE(u,v) = 0$ iff $u = cv$ for some positive $c$.
\item $\angE(u,v) = \angE(v,u)$ for all nonzero $u,v$.
\item $\angE(u,w) \leq \angE(u,v) + \angE(v,w)$ for all nonzero $u,v,w$.
\end{enumerate}
In short, after normalization the Euclidean angle becomes a bona fide
metric, namely the spherical metric. The principal observation
behind our work is that an analogous mechanism is available over
$\Fq$, with the role of the unit sphere played by the
\emph{projective space} $\mathbb{P}(\Fqn)$.

\begin{remark}
A useful geometric coincidence is the chain
$$\angE(u,v) \;\longleftrightarrow\; \bigl|\mathrm{arc}_E(u,v)\bigr|
\;\longleftrightarrow\; \bigl|\mathrm{chord}(u,v)\bigr|
\;\longleftrightarrow\; d_E(u,v),$$
which holds whenever $u, v \in S^{n-1}$. That is, the Euclidean
angle, the arc length, the chord length, and the Euclidean distance
all agree as ways to compare two unit vectors --- up to the
appropriate transformation. In contrast, over $\Fqn$ the analogous
chain breaks down (Section~\ref{ssec:differences}), and one
genuinely needs the angular notion in addition to the distance.
\end{remark}

\section{The Hamming Angle: Definition and Main Theorem}
\label{sec:main}

\begin{definition}\label{def:angleH}
For $u, v \in \Fqns$, the \emph{Hamming angle} between $u$ and $v$ is
$$\angH(u,v) \;:=\; \min_{c \in \Fqs}\, \dH(u, cv).$$
\end{definition}

Note that $\angH(u,v)$ is a nonnegative integer in $\{0,1,\dots,n\}$.
Equivalently, $\angH(u,v) = n - \max_{c \in \Fqs} \mathrm{agr}(u, cv)$,
where $\mathrm{agr}(x,y) := |\{i : x_i = y_i\}|$ is the agreement
count.

The following is our main theorem.

\begin{theorem}\label{thm:main}
The Hamming angle function $\angH \colon \Fqns \times \Fqns \to \Z$
satisfies the following three properties:
\begin{enumerate}[label=\textup{(\arabic*)}]
\item \textup{(Scalar identifiability)} $\angH(u,v) = 0$ if and only
      if $u = cv$ for some $c \in \Fqs$.
\item \textup{(Symmetry)} $\angH(u,v) = \angH(v,u)$ for all
      $u, v \in \Fqns$.
\item \textup{(Triangle inequality)}
      $\angH(u,w) \leq \angH(u,v) + \angH(v,w)$
      for all $u, v, w \in \Fqns$.
\end{enumerate}
\end{theorem}

\begin{proof}
\textbf{(1)} The Hamming distance is a metric, so
$\dH(u, cv) = 0$ if and only if $u = cv$. Since $\angH(u,v)$ is a
minimum of nonnegative quantities, it equals $0$ if and only if
$\dH(u, cv) = 0$ for some $c \in \Fqs$, i.e., $u = cv$ for some
$c \in \Fqs$.

\medskip\noindent
\textbf{(2)} Let $c_0 \in \Fqs$ be an element attaining the minimum
in $\angH(u,v)$, so that $\angH(u,v) = \dH(u, c_0 v)$. Then
$$\dH(u, c_0 v) = \wt(u - c_0 v) = \wt\!\bigl(c_0^{-1}(u - c_0 v)\bigr)
                = \wt(c_0^{-1} u - v) = \dH(v, c_0^{-1} u),$$
where the second equality uses the invariance of $\wt$ under
multiplication by the nonzero scalar $c_0^{-1}$. Since
$c_0^{-1} \in \Fqs$, we conclude
$$\angH(u,v) = \dH(v, c_0^{-1} u) \;\geq\; \min_{c \in \Fqs} \dH(v, cu)
= \angH(v, u).$$
The reverse inequality follows by the same argument starting from
$\angH(v,u)$. Hence $\angH(u,v) = \angH(v,u)$.

\medskip\noindent
\textbf{(3)} By part (2) we may equivalently write
$$\angH(u,v) = \min_{c \in \Fqs} \dH(c u, v).$$
Pick $c_1, c_2 \in \Fqs$ attaining the minima for $\angH(u,v)$ and
$\angH(v,w)$, respectively, so that
$$\angH(u,v) = \dH(c_1 u, v), \qquad
  \angH(v,w) = \dH(v, c_2 w).$$
By the triangle inequality for $\dH$,
$$\angH(u,v) + \angH(v,w) = \dH(c_1 u, v) + \dH(v, c_2 w)
                       \;\geq\; \dH(c_1 u, c_2 w).$$
Now use again the scalar-invariance of $\wt$:
$$\dH(c_1 u, c_2 w) = \wt(c_1 u - c_2 w)
                  = \wt\!\bigl(c_1^{-1}(c_1 u - c_2 w)\bigr)
                  = \wt(u - c_1^{-1}c_2\, w)
                  = \dH\!\bigl(u, c_1^{-1}c_2\, w\bigr).$$
Since $c_1^{-1}c_2 \in \Fqs$,
$$\dH(c_1 u, c_2 w) \;\geq\; \min_{c \in \Fqs} \dH(u, cw)
                          = \angH(u, w).$$
Combining the two displays gives the triangle inequality. \qedhere
\end{proof}

\section{Properties of the Hamming Angle}\label{sec:properties}

\subsection{Projective-space interpretation}\label{ssec:projective}

By definition, $\angH(u, v) = \angH(\alpha u, \beta v)$ for any
$\alpha, \beta \in \Fqs$: the rescaling is absorbed into the choice
of $c$ inside the minimum. Hence $\angH$ depends only on the
\emph{projective points} $[u], [v] \in \mathbb{P}(\Fqn)$.

\begin{proposition}\label{prop:projective-metric}
The function $\overline{\angH}([u], [v]) := \angH(u, v)$ is a
well-defined integer-valued metric on $\mathbb{P}(\Fqn)$.
\end{proposition}

Thus, just as $\angE$ is the geodesic metric on $S^{n-1}$, $\angH$ is
a natural metric on the projective space $\mathbb{P}(\Fqn)$. In
particular, $\mathbb{P}(\Fqn)$ plays the role over $\Fq$ that the unit
sphere plays over $\R$.

\begin{remark}
The metric $\overline{\angH}$ is (essentially) the \emph{projective
Hamming metric}, and the underlying quantity $\min_{c}\dH(u,cv)$ has
been studied in coding theory in the context of projective Reed--Solomon
codes and Grassmann codes (where one considers distances between
$k$-dimensional subspaces). The novelty of the present viewpoint is
the explicit angular interpretation and the resulting metric axioms,
which appear not to have been recorded together in this form.
\end{remark}

\subsection{Efficient computation of $\angH$}
\label{ssec:computation}

Although $\angH(u, v)$ is a minimum over $q-1$ scalars, it can be
computed in a single pass:

\begin{proposition}\label{prop:fast-compute}
Given $u, v \in \Fqns$, $\angH(u, v)$ can be computed in time $O(n)$
(under the standard assumption that field operations cost $O(1)$).
\end{proposition}

\begin{proof}
Partition $[n]$ according to the joint values of $(u_i, v_i)$:
\begin{itemize}
\item $N_0 := |\{i : u_i = 0 \text{ and } v_i = 0\}|$
     --- positions where $u$ and $cv$ agree for \emph{every} $c$;
\item $N_1 := |\{i : u_i \neq 0,\; v_i = 0\}|$
     --- positions where $u_i \neq (cv)_i = 0$ for every $c \in \Fqs$;
\item $N_2 := |\{i : u_i = 0,\; v_i \neq 0\}|$
     --- positions where $u_i = 0 \neq cv_i$ for every $c \in \Fqs$;
\item For $c \in \Fqs$, $A_c := |\{i : u_i, v_i \neq 0,\; u_i = c v_i\}|$
     --- positions where $u_i = (cv)_i$ for that particular $c$.
\end{itemize}
A position contributes to $\dH(u, cv)$ unless it is counted in
$N_0$ or in $A_c$. Therefore
$$\dH(u, cv) = n - N_0 - A_c, \quad\text{and}\quad
  \angH(u, v) = n - N_0 - \max_{c \in \Fqs} A_c.$$
To compute $\max_c A_c$, iterate over $i \in [n]$: whenever both
$u_i$ and $v_i$ are nonzero, increment a counter indexed by the
ratio $u_i / v_i \in \Fqs$. The maximum counter value is
$\max_c A_c$. The total work is $O(n)$.
\end{proof}

\subsection{Extremal values and small examples}\label{ssec:examples}

The smallest value $\angH(u,v) = 0$ is achieved exactly when $u, v$
span the same projective point. At the other extreme:

\begin{proposition}\label{prop:max-angle}
For $u, v \in \Fqns$,
$\angH(u, v) \leq n - 1$ \emph{unless} at every position exactly one
of $u_i, v_i$ vanishes, in which case $\angH(u, v) = n$.
\end{proposition}

\begin{proof}
If some position $i$ has $u_i \neq 0$ and $v_i \neq 0$, take
$c = u_i v_i^{-1} \in \Fqs$ in Proposition~\ref{prop:fast-compute};
then $A_c \geq 1$, so $\angH(u,v) \leq n - 1$. Otherwise,
$N_0 + N_1 + N_2 = n$ and $\max_c A_c = 0$, so
$\angH(u,v) = n - N_0$. The latter equals $n$ exactly when $N_0 = 0$,
i.e., exactly one of $u_i, v_i$ is nonzero at every position.
\end{proof}

\begin{example}[The case $q = 2$]
Since $\Fqs = \{1\}$ when $q = 2$, the definition reduces to
$\angH(u,v) = \dH(u,v)$. In other words, over $\F_2$ the angular
metric coincides with the Hamming metric on $\F_2^n\setminus\{0\}$.
The genuinely new content of $\angH$ appears for $q \geq 3$.
\end{example}

\begin{example}[A computation in $\F_3^3$]\label{ex:f3}
Let $u = (1, 2, 0)$ and $v = (1, 1, 2)$ in $\F_3^3$. Compute:
\begin{itemize}
\item For $c = 1$: $cv = (1,1,2)$ and $\dH(u, cv) = |\{2, 3\}| = 2$.
\item For $c = 2$: $cv = (2,2,1)$ and $\dH(u, cv) = |\{1, 3\}| = 2$.
\end{itemize}
Hence $\angH(u,v) = 2$. Note that
$u \cdot v = 1\cdot 1 + 2\cdot 1 + 0\cdot 2 = 3 \equiv 0 \pmod{3}$,
so $u$ and $v$ are ``orthogonal'' in the inner-product sense, yet
$\angH(u,v) = 2 < 3 = n$. This illustrates the next subsection.
\end{example}

\subsection{Differences from the Euclidean angle}
\label{ssec:differences}

We highlight two qualitative differences between $\angH$ and $\angE$.

\paragraph{(A) Inner products are not faithful detectors of right angles.}
Over $\R^n$, $u \cdot v = 0$ \emph{iff} $\angE(u,v) = \pi/2$. The
corresponding equivalence fails over $\Fq$: as in
Example~\ref{ex:f3}, one can have $u\cdot v = 0$ while $\angH(u,v)$
is strictly less than $n$. In fact, $\angH(u,v) = n$ \emph{implies}
$u \cdot v = 0$ (since by Proposition~\ref{prop:max-angle} the
supports of $u$ and $v$ are disjoint), but the converse fails.
Consequently, there is no clean way to read off ``trigonometric''
quantities from $\angH$, and the relation between the inner product
and the angle is one-way only.

\paragraph{(B) $\angH$ is statistical-distance-like.}
Recall that the statistical distance between two probability
distributions $P, Q$ is $1$ exactly when $P$ and $Q$ have disjoint
supports. By Proposition~\ref{prop:max-angle}, the analogous
statement holds for $\angH$: the maximum value $\angH(u,v) = n$ is
achieved precisely when $u, v$ have ``disjoint structure'' in the
sense that at every coordinate at most one of them is nonzero. The
Euclidean $\angE$ exhibits no such phenomenon.

\section{An Application: Projective Unique Decoding}
\label{sec:application}

We now apply $\angH$ to a concrete problem in coding theory. Our
result is an angular (or \emph{projective}) refinement of the
classical unique-decoding theorem.

\subsection{The Hamming angle to a set}
\label{ssec:angle-to-set}

\begin{definition}\label{def:angle-to-set}
For nonzero $u \in \Fqn$ and a subset $S \subseteq \Fqn$ with
$S \cap \Fqns \neq \emptyset$, define
$$\angH(u, S) \;:=\; \min_{v \in S \cap \Fqns}\, \angH(u, v).$$
\end{definition}

The next proposition records the exact relationship between
$\angH(u, S)$ and the Hamming distance $\dH(u, S)$ when $S$ is a
linear subspace. We caution that the naive identity $\angH(u, S) =
\dH(u, S)$ is \emph{false} in general: the two quantities differ when
the zero vector is the closest point of $S$ to $u$.

\begin{proposition}\label{prop:angle-vs-dist}
Let $C \subseteq \Fqn$ be a linear subspace with $C \neq \{0\}$, and
let $u \in \Fqns$. Then
$$\angH(u, C) \;=\; \min_{c \in C \setminus \{0\}}\, \dH(u, c)
\;\geq\; \dH(u, C),$$
with equality $\angH(u, C) = \dH(u, C)$ if and only if the minimum
$\dH(u, C) = \min_{c \in C} \dH(u, c)$ is attained at some nonzero
$c \in C$.
\end{proposition}

\begin{proof}
Since $C$ is a subspace, for every $v \in C\setminus\{0\}$ and every
$c \in \Fqs$ we have $cv \in C\setminus\{0\}$. Conversely, any
$w \in C\setminus\{0\}$ is of the form $1\cdot w$. Hence
\begin{align*}
\angH(u, C) &= \min_{v \in C\setminus\{0\},\; c \in \Fqs}\dH(u, cv)
            = \min_{w \in C\setminus\{0\}} \dH(u, w).
\end{align*}
Comparing this with
$\dH(u, C) = \min_{w \in C} \dH(u, w) = \min\bigl(\wt(u),\,
\min_{w \in C\setminus\{0\}}\dH(u, w)\bigr)$
gives the claim.
\end{proof}

\begin{remark}\label{rem:proximity-regime}
In the regimes of interest for proximity testing
\cite{BCKS25,GHKSS25}, the input $u$ is purported to be close to a
codeword $c \neq 0$, in particular closer to $c$ than to $0$. There,
$\angH(u, C) = \dH(u, C)$ and the two notions coincide. The careful
distinction in Proposition~\ref{prop:angle-vs-dist} is what one needs
in order to use $\angH$ as a drop-in replacement for $\dH$ in such
arguments.
\end{remark}

\subsection{Angular unique decoding}
\label{ssec:angular-decoding}

Recall that the \emph{minimum distance} of a linear code
$C \subseteq \Fqn$ is
$d := \min_{c \in C\setminus\{0\}} \wt(c)$,
which (by linearity) coincides with
$$\min_{c, c' \in C,\, c\neq c'} \dH(c, c').$$ The classical unique
decoding theorem states that, for any $u \in \Fqn$, there is at most
one codeword $c \in C$ with $\dH(u, c) < d/2$.

We prove the angular analog.

\begin{theorem}[Angular unique decoding]\label{thm:angular-unique-decoding}
Let $C \subseteq \Fqn$ be a linear code of minimum distance $d$, and
let $u \in \Fqns$. If $\angH(u, C) < d/2$, then the minimum
$\angH(u, C) = \min_{c \in C\setminus\{0\}} \angH(u, c)$
is attained at a codeword $c^\star \in C\setminus\{0\}$ that is
unique up to multiplication by an element of $\Fqs$. Equivalently,
the projective point $[c^\star] \in \mathbb{P}(C)$ is uniquely
determined by $u$.
\end{theorem}

\begin{proof}
Suppose for contradiction that $c_1, c_2 \in C\setminus\{0\}$ both
satisfy $\angH(u, c_i) < d/2$ for $i = 1, 2$, and that
$c_1 \notin \Fqs \cdot c_2$, i.e., $[c_1] \neq [c_2]$ in
$\mathbb{P}(C)$.

By the triangle inequality of $\angH$
(Theorem~\ref{thm:main}(3)),
$$\angH(c_1, c_2) \leq \angH(c_1, u) + \angH(u, c_2)
                  < \frac{d}{2} + \frac{d}{2} = d.$$
On the other hand,
$$\angH(c_1, c_2) = \min_{\alpha \in \Fqs} \dH(c_1, \alpha c_2)
                = \min_{\alpha \in \Fqs} \wt(c_1 - \alpha c_2).$$
For each $\alpha \in \Fqs$, $c_1 - \alpha c_2 \in C$ by linearity,
and $c_1 - \alpha c_2 \neq 0$ because $[c_1] \neq [c_2]$ (otherwise
$c_1 = \alpha c_2$ would put $c_1$ in $\Fqs \cdot c_2$). Therefore
each $c_1 - \alpha c_2$ is a \emph{nonzero} codeword, and
$\wt(c_1 - \alpha c_2) \geq d$. Taking the minimum over $\alpha$
gives $\angH(c_1, c_2) \geq d$, contradicting the previous display.
\end{proof}

\begin{remark}\label{rem:vs-classical}
Theorem~\ref{thm:angular-unique-decoding} is the projective analog
of classical unique decoding. The classical theorem fixes a
codeword $c^\star$, whereas the angular version fixes only the
\emph{direction} $[c^\star]$. Both theorems share the same radius
$d/2$, and indeed the proof reduces to applying the triangle
inequality of an appropriate metric --- $\dH$ in the classical case,
$\angH$ in the projective case. In this sense
Theorem~\ref{thm:angular-unique-decoding} is not a strengthening of
the classical theorem but rather its natural projective sibling.
\end{remark}

\begin{corollary}\label{cor:list-size-one}
Under the hypothesis of Theorem~\ref{thm:angular-unique-decoding}, the
projective list-decoding list
$$L_{\mathrm{proj}}(u, \rho) := \bigl\{[c] \in \mathbb{P}(C) :
\angH(u, c) < \rho \bigr\}$$
has size at most $1$ for any $\rho < d/2$.
\end{corollary}

\subsection{Relation to the proximity-gap problem}
\label{ssec:relation-proximity}

The proximity-gap phenomenon \cite{BCKS25,GHKSS25} concerns
\emph{affine} lines in $\Fqn$: given $f, g \in \Fqn$, one studies the
line $L_{f,g} := \{f + zg : z \in \Fq\}$ and asks whether the points
on $L_{f,g}$ are uniformly close to a code $C$ or, on the contrary,
uniformly far from it. The relevant codewords come in correlated
pairs $(c_0, c_1)$, with $c_0 + z c_1$ tracking $f + z g$ on a common
agreement set.

The function $\angH$ does not directly capture this affine geometry;
it is by design invariant under scalar multiplication only in
\emph{each argument separately}, whereas a proximity-gap statement
needs the joint linear structure of $(f, g)$. However, the two
frameworks meet at the following points.

\begin{itemize}
\item \emph{Both produce decoding radii of the form $d/2$ or $J(\delta)$.}
Theorem~\ref{thm:angular-unique-decoding} corresponds to the
``trivial'' unique-decoding-radius regime in which the proximity loss
$\varepsilon^\star$ is essentially $0$. Up to the Johnson radius
$J(\delta) = 1 - \sqrt{1-\delta}$, recent work \cite{BCKS25} achieves
the same $\varepsilon^\star = 0$ for affine lines, but at the price
of $O(n)$ exceptional values of $z$.

\item \emph{Correlated agreement is a 2-vector projective statement.}
A pair of vectors $(u_0, u_1)$ has $\delta$-correlated agreement with
a code $C$ if there exist $c_0, c_1 \in C$ and $T \subseteq [n]$ with
$|T| \geq (1-\delta)n$ such that $u_i$ and $c_i$ agree on $T$ for
each $i$. This is precisely the statement that the
\emph{$2$-dimensional} subspace $\langle u_0, u_1\rangle \subseteq
\Fqn$, modulo $C \oplus C$, has a representative of small weight ---
a higher-dimensional generalization of the $1$-dimensional projective
notion captured by $\angH$.

\item \emph{The triangle inequality of $\angH$ provides a uniqueness
obstruction.} As in Theorem~\ref{thm:angular-unique-decoding}, the
triangle inequality immediately forbids two distinct projective
codewords from being simultaneously $\rho$-close to the same input
when $2\rho < d$. The corresponding obstruction in the affine setting
underlies the unique-decoding regime of all known proximity-gap
proofs.
\end{itemize}

A natural direction for future work is to define a \emph{projective}
proximity-gap problem in the $\angH$ framework --- i.e., projective
lines in $\mathbb{P}(\Fqn)$ rather than affine lines in $\Fqn$ --- and
to determine whether the $\angH$ triangle inequality and the
projective-space geometry yield improvements in the unique-decoding
regime analogous to those of \cite{BCKS25}. We leave this to future
investigation.

\section{Concluding Remarks}\label{sec:conclusion}

We have introduced an angular metric $\angH$ on nonzero vectors over
$\Fq$, defined by minimizing the Hamming distance over scalar
multiples of one argument. The function $\angH$ satisfies the
metric axioms up to scalar equivalence and descends to a genuine
integer-valued metric on the projective space $\mathbb{P}(\Fqn)$. As
an application, we showed that the triangle inequality of $\angH$
gives a clean angular analog of classical unique decoding
(Theorem~\ref{thm:angular-unique-decoding}).

Several directions invite further study.
\begin{enumerate}[label=(\arabic*)]
\item \emph{Higher-dimensional angles.} The angle between a vector
and a $k$-dimensional projective subspace, or between two such
subspaces, is the natural Grassmannian counterpart of $\angH$. This
should connect $\angH$ to existing work on Grassmann codes and
subspace designs \cite{GHKSS25}.

\item \emph{Affine $\angH$.} Adapting $\angH$ to capture
\emph{affine} lines (in particular, sets of the form
$\{f + zg : z \in \Fq\}$) would provide an angular language for the
proximity-gap problem itself, not just for its projective shadow.

\item \emph{Other metric vector spaces.} The construction of $\angH$
from $\dH$ uses only the scalar-invariance of $\wt$. Any
metric vector space whose norm is invariant under nonzero scalars
admits an analogous angular metric. It would be interesting to study
such constructions over rings or for non-Hamming metrics
(rank metric, Lee metric, etc.).
\end{enumerate}


\end{document}